\documentclass{article}
\usepackage{amsthm}
\usepackage{textcomp}
\input amssym.def
\input amssym.tex
\def\ben{\begin{equation}}
\def\een{\end{equation}}

\def\bea{\begin{eqnarray}}
\def\eea{\end{eqnarray}}

\newtheorem{thm}{Theorem}[section]
\newtheorem{lem}[thm]{Lemma}

\usepackage{graphicx}
\begin{document}

\hfuzz=100pt
\title{Non-smooth geometry and collapse of flexible structures under smooth loads}
\author{G. Cs\'anyi, \\
{\em Department of Engineering}\\
{\em  Cambridge University}\\
{\em Trumpington Street, Cambridge CB2 1PZ, U.K.}\\
{\em and}  \\
G. Domokos \\
{\em Department of Mechanics, Materials, and Structures }\\
{\em Budapest University of Technology and Economics}\\
{\em M\H uegyetem rkpt.3 }\\
{\em Budapest 1111, Hungary} }

\date{\today}       
 \maketitle

\begin{abstract}
While static equilibria of flexible strings subject to various
load types (gravity, hydrostatic pressure, Newtonian wind) is well understood textbook material,
the combinations of the very same loads can give rise to complex
spatial behaviour at the core of which is the unilateral material
constraint prohibiting compressive loads. While the effects of
such constraints have been explored in 
optimisation problems involving straight cables,
the geometric complexity of physical configurations has
not yet been addressed. Here we show that flexible strings subject to combined smooth loads may not have smooth
solutions in certain ranges of the load ratios.
This non-smooth phenomenon is closely related to the collapse geometry of inflated
tents. After proving the nonexistence of smooth solutions for
a broad family of loadings we identify two alternative, critical geometries immediately
preceding the collapse. We verify these analytical results by dynamical
simulation of flexible chains as well as with simple table-top experiments with
an inflated membrane.
\end{abstract}

\section{Introduction}

The geometry and equilibrium configuration of flexible strings has been studied ever since the
classical papers  on the catenary by Bernoulli \cite{Bernoulli}, Leibniz \cite{Leibniz} and Huygens \cite{Huygens} in 1691.
In addition, Jacob Bernoulli discovered that the same curve is the solution to the problem of the shape of a sail under Newtonian parallel wind \cite{BernoulliJac}.
Classical textbooks \cite{Beer1} extend the description of stable equilibria of flexible strings to various  load types (e.g. gravity, wind or hydrostatic pressure) as analytical solutions of boundary value problems associated with ordinary differential equations. Such stable configurations can only fail due to intrinsic material nonlinearities \cite{Beer2}. On the other hand, elastic structures have additional failure modes due to global geometrical instabilities, as described first by Euler in 1744 \cite{Euler}.
More recently, the nonlinear behaviour of strings has been discussed by
Antman \cite{Antman} who describes three types of loading: gravity ($g$), suspension bridge load ($p$) and hydrostatic pressure
($h$). In addition we also consider the parallel newtonian wind ($w$) discussed by Jacob Bernoulli and  henceforth we refer to these as \em classical \rm loads.  In case of classical loads
Antman identifies multiple, smooth solutions for nonlinearly elastic strings. One of these solutions is stable
under tension and there exist multiple, smooth compressive solutions which are unstable.
Here we look at the very same four classical loading cases and show that their linear combination,
resulting in \em non-classical \rm loads,  can lead to unexpected and complex
phenomena. While the governing ODEs can be readily derived, their solutions are nontrivial.
In fact, for these non-classical loads we show that for certain finite ranges of
the load ratios $p/g,h/g,w/g$,  {\em no smooth solution exists}.

At the very core of these unexpected phenomena is the \emph{unilateral constraint} that strings can not carry compressive loads. It is well known that external
unilateral constraints can result in highly complex patterns in bifurcation diagrams as well as in spatial configurations \cite{HolmesDomokos, Pellegrino}. One notable example is a compressed, twisted elastic fiber
with self-contact, leading to the well-known mechanical
model for the geometry of supercoiled DNA molecules \cite{Coleman}. The geometry of strings in the presence of external  unilateral constraints
has also been investigated \cite{Khodja}, however, 
much less is known about the global geometrical consequences of {\em internal} (material) constraints.
The latter
has only been explored in the context of optimisation problems for straight cables \cite{Kanno, Patriksson}, the full geometric complexity in the spatial configurations
of flexible strings arising as a consequence of unilateral material behavior has not yet been described. The main reason for this is that under the above-mentioned 
classical loads the entire length of the string is either under tension or it forms a (physically irrelevant), purely compressed arch. The combined, non-classical
loads, scrutinised in this paper, can lead to the tension vanishing either point-wise (resulting in smooth tensile
segments joined by kinks) or the tension vanishing identically (at nonzero load parameter), resulting in spatially complex patterns \cite{moviepg, moviehg}.
As we will show, due to the unilateral constraint in the material behavior 
the phase space of the governing ODE is separated into disjoint, invariant domains.
If the boundary conditions admit solutions within one single domain then this solution is smooth 
(this is the case for classical loads), if however, the solution  involves trajectory-segments
in more than one domain then we see non-smooth,  spatially complex shapes.
We will illustrate using numerical simulation 
that beyond presenting surprising geometric features in static equilibrium configurations, strings under non-classical loads
display rich dynamical behaviour.  

In Section \ref{s:background} we derive the governing equations and establish the general conditions for the
non-existence of smooth solutions. In Section \ref{s:examples} we show how these conditions apply in case of specific non-classical
loads and what the critical ranges of the load parameters are. In Section \ref{s:numerical} we show numerical simulations illustrating non-smooth, complex spatial
behavior in the critical ranges of the load parameter. Section \ref{s:experiments} briefly describes our simple table-top experiment
and compares the results to the simulations. In Section \ref{s:conclusions} we draw conclusions and discuss
possible generalisations and applications.

\section{Mathematical background}\label{s:background}
\subsection{Governing equations}

We assume a flexible (ie. zero bending stiffness), inextensible string in a
vertical plane,
parametrized by the arc length $s$ with angle of slope $\alpha(s)$ measured with respect to the horizontal axis $x$, and tension
$T(s)$. We start with the formula for the \em hoop stress \rm  which, according to Gordon \cite{Gordon} originates
in Mariotte's formula for cylindrical vessels:
\ben \label{Mariotte}
f_n=\frac{T}{R}
\een
describing the relationship between the radial force (normal pressure) $f_n$, the tension $T$ and
the radius of curvature $R$. Equation (\ref{Mariotte})
may be rewritten as
\ben \label{aa1}
\kappa= \dot \alpha = \frac{f_n}{T}
\een
where $\kappa$ denotes the curvature and $\dot{()}=d/ds$.
By supplementing the radial (normal) equilibrium equation (\ref{aa1}) with the (trivial) tangential equilibrium
\ben \label{aa2}
\dot T = f_t
\een
and noting that the flexural rigidity did not enter into these equations, we see that (\ref{aa1}-\ref{aa2})
provide the full description of the equilibrium of flexible strings. This seems to be of some historical interest as Mariotte's formula
appears to be from 1680, whereas the
equations for the string were first published (simultaneously) by Johann Bernoulli \cite{Bernoulli}, Leibniz \cite{Leibniz}
and Huygens \cite{Huygens} over a decade later.

\subsection{Cartesian frame and classical loads}
If we are also interested in the spatial configuration $x(s),y(s)$ in the global $[xy]$ orthogonal frame
then we have to supplement (\ref{aa1}-\ref{aa2}) by

\bea \label{y}
\dot y & = & \sin(\alpha) \\
\label{x}
\dot x & = & \cos(\alpha), 
\eea
(expressing inextensibility) with corresponding boundary conditions

\bea \label{b1}
x(0) & = & 0  \\
\label{b2}
y(0) & = & 0 \\
\label{b3}
x(L) & = & x_0 \\
\label{b4}
y(L) & = & y_0
\eea
which, together with
the governing equations define a well posed boundary value problem (BVP) (\ref{aa1}-\ref{b4}).
We will investigate one-parameter families of these BVPs, where the control parameter
$\lambda$ will correspond to the load, either as load intensity or as the ratio
of two load intensities.
 Without loss of generality we assume $x_0 >0$ and for later reference
and we introduce the ``boundary slope'' $\alpha_0$ as
\ben \label{alfa0}
-\frac{\pi}{2}\leq \alpha_0=\arctan\left(\frac{y_0}{x_0}\right) <\frac{\pi}{2}.
\een

Before continuing, we show that the connection to the Cartesian representation is
straightforward. Indeed, since
\ben
\cos(\alpha)=\frac{1}{\sqrt{1+y'^{2}}}, \quad \dot\alpha=\frac{y''}{(1+y'^{2})^{3/2}}, \quad()'=d/dx,
\een
and the horizontal component $H$ of the tangential force $T$ is obtained as
\ben \label{force}
H  =  T\cos(\alpha),
\een
equations (\ref{aa1})-(\ref{aa2}) may be rewritten as
\bea \label {xy1}
y'' & = & \frac{f_n}{H\cos^2(\alpha)} \\
\label {xy2}
H' & =  & f_t-f_ny'.
\eea
If we now consider a vertical load $g$, the normal pressure  and tangential load $f_n, f_t$ can be expressed as
\ben \label{gravityload}
f_n = g\cos(\alpha)\qquad f_t=g\sin(\alpha).
\een 
By substituting (\ref{gravityload}) into (\ref{xy1})-(\ref{xy2}) we have
\bea \label{g2}
y'' & = & \frac{g}{H\cos(\alpha)} \\
H' & = & 0
\eea
from which the classical ODE for the catenary
\ben \label{gravity}
y''=\frac{g}{H}\sqrt{1+y'^{2}}
\een
immediately follows. 
If instead of the dead weight $g$ we substitute the ``bridge load'' $p=g\cos(\alpha)$ then we get
\ben \label{bridgeload}
f_n = p\cos^2(\alpha) \qquad f_t=p\cos(\alpha)\sin(\alpha)
\een 
and by substituting into (\ref{xy1})-(\ref{xy2}) we arrive at the other well-known ODE
\bea \label{bridge}
y'' & = & \frac{p}{H} \\
H' & = & 0
\eea
yielding the parabola as its solution. The ``Newtonian'' model of vertical wind load with intensity $w/2$ assumes elastic
collisions between the incoming particles and the string, resulting in loading that is normal
to the string's tangent:
\ben \label{windload}
f_n = w\cos ^2(\alpha), \qquad f_t=0,
\een 
so substitution into (\ref{force})-(\ref{xy1}) yields
\bea \label{w2}
y'' & = & \frac{w}{T\cos(\alpha)} \\
T' & = & 0
\eea
which shows complete analogy to (\ref{g2}) and also yield the classical ODE for the catenary
\ben \label{wind}
y''=\frac{w}{T}\sqrt{1+y'^{2}}.
\een
Thus the physical shape in the $[xy]$ plane of the string under gravity and vertical, Newtonian wind is
identical, however, in the first case the horizontal component of the tension is constant
while in the second case the magnitude of the tension itself is constant, so in the
$[\alpha,T]$ phase space of (\ref{aa1}-\ref{aa2}) the solutions appear as different trajectories.
In case of hydrostatic pressure $h$ we have
\ben \label{pressureload}
f_n = h, \qquad f_t=0,
\een 
so substitution into (\ref{xy1})-(\ref{xy2}) yields
\bea \label{pressure}
y'' & = & \frac{h}{H}(1+y'^2) \\
H' & = & -hy',
\eea
for which the solution is a circular arc. 

\subsection{Critical solutions of the IVP and their interpretation in the phase plane}\label{ss:critical}

In general, the  load components $f_n, f_t$ may depend on the spatial location, however, here we only treat loads
which  depend only on the slope $\alpha$, i.e. we have 
\ben \label{loads}
f_n=f_n(\alpha) \qquad f_t=f_t(\alpha).
\een 
We remark that all four  classical loads
(gravity (\ref{gravityload}), bridge load (\ref{bridgeload}), Newtonian wind (\ref{windload}) and hydrostatic pressure (\ref{pressureload})) satisfy (\ref{loads}). This implies that the system (\ref{aa1}-\ref{aa2}) 
becomes autonomous and can be integrated (at least numerically)
as an initial value problem (IVP) to find $\alpha(s)$ and $T(s)$. We can make several simple
observations about the phase portrait of (\ref{aa1}-\ref{aa2}). First we note that
the line $T=0$ separates the $[\alpha,T]$ plane into two invariant half-planes
and trajectories for sufficiently small values of $T$ run almost parallel to $T=0$
except for isolated points where $f_n(\alpha)=0$.
Regarding the latter, we  observe that the IVP associated with (\ref{aa1}-\ref{aa2})
is solved by the simple Ansatz
\ben \label{ansatz}
\alpha(s)=\alpha^\star_i=constant, \qquad i=1,2,\dots n
\een
where $ \alpha^\star_{i} $ are the real roots 
\ben \label{critical}
f_n(\alpha)=0
\een
and we call (\ref{ansatz}) \em critical \rm solutions of (\ref{aa1}-\ref{aa2}). 
These solutions appear as straight, vertical lines in the $[\alpha,T]$ phase space
of (\ref{aa1}-\ref{aa2}) (see Figure~\ref{phasefig}). Since the phase space is a unique representation of
solutions of the IVP (i.e. trajectories cannot intersect), this implies that if
$n>1$ then the phase space is partitioned
into invariant domains (vertical semi-infinite stripes), defined by consecutive roots $\alpha^\star_i,\alpha^\star_{i+1}$ and
the $T=0$ axis. All non-critical, smooth
solutions are  trapped in one of these semi-infinite stripes, i.e. their slopes
will be bounded both from below and from above.
Since any smooth solution of (\ref{aa1}-\ref{aa2}) will be trapped inside one of
the invariant domains,  the curvature $\kappa(s) = \dot \alpha(s)$ along these solutions
is bounded away from zero, i.e. such solutions can not have inflection points
in the physical $[xy]$ space. Non-smooth solutions may exhibit \it kinks \rm
which appear in the physical $[xy]$ space
as two straight lines, corresponding to critical solutions (satisfying (\ref{critical})) joined at a vertex.
Such kinks may only occur at vanishing tension, so they appear at the bottom
of the boundary of the invariant stripes in the phase space: here  kinks correspond to two 
neighboring critical solutions $ \alpha^\star_{i}, \alpha^\star_{i+1}$, joined by
the horizontal $[\alpha^\star_{i}, \alpha^\star_{i+1}]$ segment of the $\alpha$-axis.
We call such kinks \em mathematical \rm as opposed to \em physical kinks \rm where
segments with non-vanishing curvatures are joined at a vertex. Physical kinks can be
only sustained by strings with finite bending stiffness and/or self-contact of the string which
are not included in our mathematical model. Nevertheless,
mathematical kinks are the key elements in the description
of our model and the interpretation of the physical processes.
In numerical and physical experiments
we expect to see smooth solutions to display first mathematical kinks which subsequently evolve
into physical kinks.

Of special interest is the case when
the boundary slope $\alpha_0$ (defined in (\ref{alfa0})) is trapped in one of the invariant stripes, i.e.
\ben \label{bracket}
\exists k: \alpha^\star_k \leq \alpha_0 \leq \alpha^\star_{k+1}
\een
If (\ref{bracket}) is fulfilled  and the width $\delta_k$ of the stripe satisfies
\ben \label{delta}
\delta_k =\alpha^\star_{k+1}-\alpha^\star_{k} < \pi,
\een 
then we call $\alpha^\star_k, \alpha^\star_{k+1}$ the \em relevant roots \rm of (\ref{critical})
and in the next subsection we show that the existence of relevant roots excludes smooth solutions
for a sufficiently long string. 

\begin{figure}[here]
\begin{center}
\includegraphics[width=100 mm]{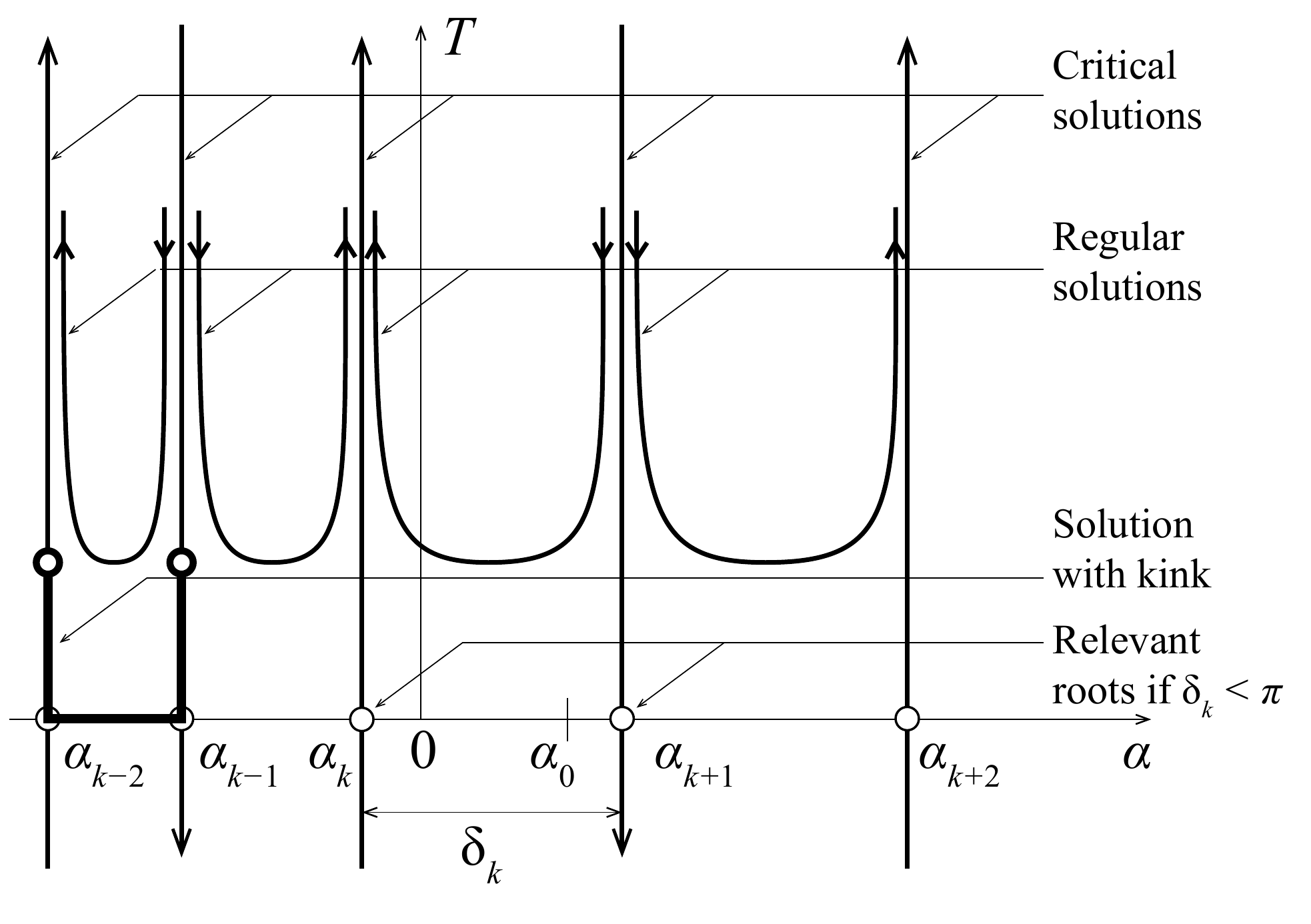}
\end{center}
\caption{Solutions in $[\alpha,T]$ phase space. Straight, vertical lines represent critical solutions corresponding to the roots of (\ref{critical}). Regular solutions are trapped in the 
semi-infinite stripes between critical solutions and the $T=0$ axis. Relavant roots exist if two critical solutions $\alpha^\star_k, \alpha^\star_{k+1}$ are bracketing the boundary slope $\alpha_0$
and  $\delta_k =\alpha^\star_{k+1}-\alpha^\star_{k} < \pi$. Kinks are produced by two neighboring critical solutions $\alpha^\star_i, \alpha^\star_{i+1}$ joined at a vertex,
the latter corresponds to the interval $[\alpha^\star_i, \alpha^\star_{i+1}]$ on the $\alpha$ axis.}\label{phasefig}
\end{figure}

\subsection{Nonexistence of smooth BVP solutions}

Based on the observations in the previous subsection regarding solutions of the IVP
(\ref{aa1}-\ref{aa2}) we can draw some immediate conclusions concerning the solutions of the BVP (\ref{aa1}-\ref{b4}).
Since smooth IVP solutions can not have inflection points in the physical $[xy]$ space, neither can smooth BVP
solutions exhibit such points. Thus, in a one-parameter ($\lambda$) BVP where there exist two parameter values $\lambda_1,\lambda_2$ for both of which the 
BVP has unique, smooth,  stable solutions along which the curvature $\kappa = \dot \alpha$
has uniform sign and this sign is opposite for $\lambda_1$ and $\lambda_2$, we can claim that those two configurations
can not be connected via a continuous family of smooth BVP solutions. Such is the case if we choose $\lambda$ as any of the load ratios $p/g, w/g,h/g$ and assume  downward direction for $g$ and an 
upward direction for $p,w,h$ defined in (\ref{bridgeload}), (\ref{windload}) and (\ref{pressureload}), respectively.
This observation implies that in a one-parameter BVP associated with non-classical loads, if we vary the control parameter continuously between its extreme values
(corresponding to classical loads with opposite directions), we will observe some non-smooth phenomenon which may
 manifest itself  either as a dynamic jump between equilibria or by the existence of non-smooth equilibria.
We will not only show that both possibilities exist, we will also derive the exact range of the control
parameter $\lambda$ where non-smooth phenomena occur. In addition we will also describe the limiting
geometries of BVP solutions at the border of this critical parameter range.

We start with a simple geometric observation about self-intersecting solutions.
Let $\mathbf{r}(s)=[x(s),y(s)]^T$ be a smooth curve, where  $s\in[0,L]$ ($L>1$) denotes the arclength
and we have the boundaries $\mathbf{r}(0)=[0,0]^T$, $\mathbf{r}(L)=[1,0]^T$. Then we have

\begin{lem}\label{lem:selfx}
\rm If $\nexists s^\star\in[0,L]$ such that $\dot\mathbf{r}(s^\star)=[1,0]^T$ then $\mathbf{r}(s)$ is self-intersecting.
\end{lem}

\begin{proof}
We start by noting that at any local extremum of $y(s)$ we have
$\dot\mathbf{r}(s)=[\pm 1,0]^T$, and then sketch a proof of the Lemma by contradiction. 
The denial of Lemma \ref{lem:selfx} is equivalent to the statement that all extrema of $y(s)$ are associated with $\dot x=-1$, however, $\mathbf{r}(s)$ has no self-intersection.
Let   $y(s)$ assume its global minimum and maximum in $s\in[0,L]$ at $s_1$ and $s_2$ respectively and assume without loss of generality that $s_1<s_2$.
Then the segment  $s\in [s_1,s_2]$ of $\mathbf{r}(s)$ will separate the stripe $y \in [y(s_1),y(s_2)]$ into two disjoint domains (cf. Figure \ref{figlemma}).
Depending on which of the boundaries $\mathbf{r}(0), \mathbf{r}(L)$ is contained in which of the disjoint domains,  either  $\mathbf{r}(0)$ and $\mathbf{r}(s_1-\epsilon)$ or
$\mathbf{r}(s_2+\epsilon)$ and $\mathbf{r}(L)$, or both pairs are separated by the segment  $s\in [s_1,s_2]$ of $\mathbf{r}(s)$,
i.e. the curve can only be connected smoothly by self-intersection. 
\end{proof}

\begin{figure}[here]
\begin{center}
\includegraphics[width=110 mm]{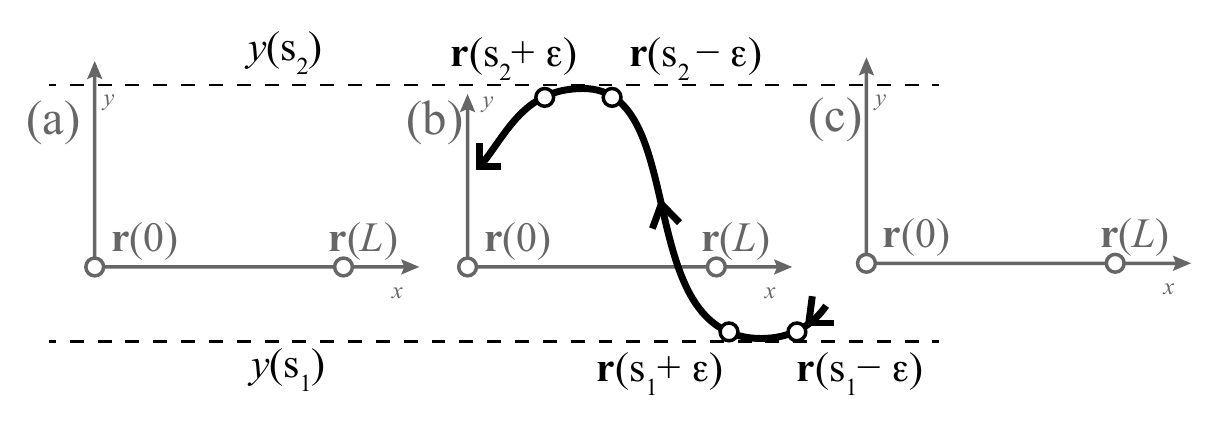}
\end{center}
\caption{Illustration for the proof of Lemma~\ref{lem:selfx}. Segment $s\in [s_1,s_2]$ of $\mathbf{r}(s)$ between the global minimum and maximum separates the stripe $y \in [y(s_1),y(s_2)]$ into two disjoint domains.
Observe that either  $\mathbf{r}(0)$ and $\mathbf{r}(s_1-\epsilon)$ (case a) or
$\mathbf{r}(s_2+\epsilon)$ and $\mathbf{r}(L)$ (case c), or both pairs (case b) are separated the segment  $s\in [s_1,s_2]$ of $\mathbf{r}(s)$.
}\label{figlemma}
\end{figure}

We proceed to formulate our main claim, and as a preparation we rotate the global $[xy]$ frame by
\ben \label{symmetry}
\varphi=0.5(\alpha^\star_k+\alpha^\star_{k+1})
\een
and we refer to the coordinates in the rotated system as $\bar x,\bar y,\bar \alpha$ etc., so we have
\ben
\bar\alpha^\star_k = -\bar \alpha^\star_{k+1}.
\een
Now we can formulate 

\begin{thm}\label{thm:smoothselfx}
If (\ref{critical}) has relevant roots and $L > \bar x_0/\cos(\bar \alpha^\star_k)$
then the BVP (\ref{aa1}-\ref{b4}) has no smooth solution free of self-intersection.
\end{thm}

\begin{proof}
We again prove the statement by contradiction, i.e. we start by assuming that
a smooth, non self-intersecting solution $\bar \alpha (s)$ exists. The lack of self-intersection implies, via Lemma~\ref{lem:selfx},
that 
\ben
\exists s^\star \mbox{ such that  }\alpha(s^\star)=\bar \alpha_0,
\een
i.e. any non self-intersecting smooth solution $\bar \alpha(s)$ is confined to
\ben \label{stripe}
\bar \alpha(s) \in (\bar\alpha^\star_k , \bar \alpha^\star_{k+1}).
\een
from which (due to the symmetrisation (\ref{symmetry}) and due to
the restriction $\delta_k<\pi$) it also follows that
\ben
\cos(\bar\alpha(s)) > \cos(\bar\alpha^\star_k) =\cos(\bar \alpha^\star_{k+1})>0
\een
We can now express the arclength:
\ben
L=\int_{s=0}^{L}ds=\int_{x=0}^{\bar x_0} \frac{d\bar x}{\cos(\bar \alpha)}< \frac{\bar x_0}{\cos(\bar \alpha^\star_{k})}.
\een
Since this contradicts the condition of the Theorem, our proof by contradiction is completed: we showed that
if (\ref{critical}) has relevant roots then the BVP has no smooth solutions
without self-intersections if the length $L$
is above the critical value indicated in the Theorem.
\end{proof}

\section{Analytical examples}\label{s:examples}

\subsection{Classical loads: no relevant roots}
We describe four well-known load types:
\begin{enumerate}
\item Gravity, with uniform weight $g$ per unit arc length
\item ``Bridge load'' with uniform vertical load $p$ per unit horizontal projection
\item Newtonian parallel, vertical wind $w$ per unit of horizontal projection
\item Hydrostatic pressure with uniform load $h$ per unit arc length, normal to the string's tangent
\end{enumerate}
and in each case we set $y_0=0$. As we already mentioned in the Introduction, three load types ($g,p,h$)
have been described in detail by Antman \cite{Antman}.
The equations for all four mentioned load-types are very well known and can be readily derived
from (\ref{aa1}-\ref{aa2}). By substituting any of (\ref{gravityload}),(\ref{bridgeload}),(\ref{windload}),(\ref{pressureload}) into (\ref{critical}) we find that 
in the interval $[-\pi,\pi]$ for gravity, bridge
and wind loads we get two roots, at distance $\delta_k=\pi$ and in case of the hydrostatic
pressure there are no roots (see Table \ref{t:1}, lines 1-4.) According to the
condition $\delta_k>\pi$ in Theorem~\ref{thm:smoothselfx}, the two roots at distance
$\pi$ in case of the vertical loads do not impose a constraint on the length $L$, so we
can expect smooth solutions. However, these roots can be physically interpreted by
saying that under these loads the string is either everywhere vertical or it has
no vertical tangent. This statement agrees with Proposition 2.25 (p592) of Antman \cite{Antman}.

\subsection{Coupled loads: existence of relevant roots}
We now investigate \em coupled loads \rm, i.e. when either $p,w$ or $h$ is acting opposite
to the gravity $g$. For the combined loads we introduce
\bea
\pm \theta_p & = & \arccos(g/p),\hspace{1 cm} |\theta_p|<\pi   \\
\pm \theta_w & = & \arccos(g/w),\hspace{1 cm} |\theta_w|<\pi   \\
\pm \theta_h& = & \arccos(h/g),\hspace{1 cm} |\theta_h|<\pi 
\eea
and as shown in lines 5-7 of Table \ref{t:1}, (\ref{critical}) may have multiple roots in $[-\pi,\pi]$, with distance $\delta_k<\pi$. 
\begin{table}[h]
	\centering
		\begin{tabular}{|l|l|l|l|l|}
		\hline
			&Load type & $f_n=0$ & i & $\alpha^\star_{i}$     \\
			\hline
			\hline
			1& $g$ & $g\cos(\alpha)=0$ & 1 & $-\pi/2$   \\
			\hline
			 & $ $ &                  & 2 & $+\pi/2$   \\
			\hline
			2& $p$ & $p\cos^2(\alpha)=0$ & 1 & $-\pi/2$   \\
					\hline
			 & $ $ &                  & 2 & $+\pi/2$   \\
			 \hline
			 3& $w$ & $w\cos^2(\alpha)=0$ & 1 & $-\pi/2$   \\
					\hline
			 & $ $ &                  & 2 & $+\pi/2$   \\
			 \hline
			4&$h$ & $h=0$ & - & -   \\
			\hline
			\hline
			5&$g-p$ & $\cos(\alpha)(g-p\cos(\alpha))=0$  & 1 & $-(\pi-\theta_p)$ \\
			\hline
			 &      &                                       & 2 & $-\pi/2$ \\
			 \hline
			 &      &                                        & {\bf 3} & $-\theta_p$ \\
			 \hline
			 &      &                                       & {\bf 4} & $+\theta_p$ \\
			 \hline
			 &      &                                       & 5 & $+\pi/2$ \\
			 \hline
			 &      &                                       & 6 & $+(\pi-\theta_p)$ \\
	\hline
			\hline
			6&$g-w$ & $\cos(\alpha)(g-w\cos(\alpha))=0$  & 1 & $-(\pi-\theta_w)$ \\
			\hline
			 &      &                                       & 2 & $-\pi/2$ \\
			 \hline
			 &      &                                        & {\bf 3} & $-\theta_w$ \\
			 \hline
			 &      &                                       & {\bf 4} & $+\theta_w$ \\
			 \hline
			 &      &                                       & 5 & $+\pi/2$ \\
			 \hline
			 &      &                                       & 6 & $+(\pi-\theta_w)$ \\
	\hline
			\hline
			7&$g-h$ & $g\cos(\alpha)-h=0$ & 1 & $-\theta_h$\\
			\hline		
			 &      &                     & 2 & $+\theta_h$ \\
			\hline

				\end{tabular}
	\caption{Critical solutions for the loads $g,p,w,h$ and the combined loads $(g-p)$,$(g-w)$,$(g-h)$. Observe that the ordering of the 6 roots in case of the gravity-bridge and the gravity-wind  load combinations (lines 5-6) depends on the absolute value of $\theta_p$ and $\theta_w$, respectively. The table shows the ordering for $\theta_p<\pi/2$, i.e. $g/p>0, g/w>0$. Here the
	distance $\delta_k=2\theta_p<\pi$, so the conditions of Theorem~\ref{thm:smoothselfx} are fulfilled.}
	\label{t:1}
\end{table}

Using the results in Table \ref{t:1} we can formulate three corollaries for these special load types.
We observe that the relevant roots are symmetric with respect to zero, so in the symmetrisation (\ref{symmetry}) we have $\varphi=0$. 
Based on Theorem~\ref{thm:smoothselfx} and Table \ref{t:1} we now have
\bigskip

\noindent \bf Corollary 1 \rm If the load is a linear combination of gravity $g$ and ``bridge load'' $p$ 
(cf. Table \ref{t:1}, line 5) then
the BVP (\ref{aa1}-\ref{b4}) has no smooth solutions if 
\begin{displaymath}
\frac{L}{x_0} > \frac{p}{g} > 1
\end{displaymath}

\noindent \bf Corollary 2 \rm If the load is a linear combination of gravity $g$ and ``wind load'' $w$ 
(cf. Table \ref{t:1}, line 6) then
the BVP (\ref{aa1}-\ref{b4}) has no smooth solutions if 
\begin{displaymath}
\frac{L}{x_0} > \frac{w}{g} > 1
\end{displaymath}

\noindent \bf Corollary 3 \rm If the load is a linear combination of gravity $g$ and hydrostatic pressure $h$ (cf. Table \ref{t:1}, line 7) then
the BVP (\ref{aa1}-\ref{b4}) has no smooth solutions if 
\begin{displaymath}
\frac{x_0}{L} < \frac{h}{g} < 1
\end{displaymath}
These corollaries can be obtained from Theorem~\ref{thm:smoothselfx} by substituting the values from Table \ref{t:1} for the specific
load combinations. We can observe that in all three cases as the load ratios $g/p,g/w,h/g$ pass through zero, the corresponding roots
$\pm \theta_p,\pm \theta_w, \pm \theta_h$ cease to be relevant because $\delta_k>\pi$.

One can describe a bifurcation problem in all three cases, and for the bridge and wind loads these coincide, so henceforth we do not treat 
the wind load case separately. The natural bifurcation parameters for the bridge load and hydrostatic cases are
$p/g$ and $h/g$, respectively. In Figure \ref{mainfig} we show the bifurcation diagrams, plotting the parameter values versus the roots of (\ref{critical}). We also include the corresponding physical configurations
in the $[xy]$ plane as well as the corresponding phase portraits of (\ref{aa1}-\ref{aa2}) in the $[\alpha,T]$
 space.

\subsection{Limiting geometries}

The corollaries of the previous subsection indicate the intervals of the load-ratios where
smooth solutions do not exist. The first natural question is to establish the limiting
geometries corresponding to equilibrium configurations at the endpoints of these parameter intervals.

\subsubsection{Bridge and wind coupled with gravity}

In case of the gravity-bridge load combination we can observe that at the lower end of the
critical interval of the load ratio, as $p \to g$ from above, the two relevant roots $\pm \theta_p$
of (\ref{critical}) approach simultaneously each other and zero. Also,
we note that critical solutions at $\pm \pi/2$ exist, corresponding to vertical
string segments. At the point of vanishing tension
a mathematical kink (cf. subsection \ref{ss:critical}) will appear between the horizontal and vertical critical solution
segments. Consequently, we expect the limiting geometry
of the smooth string to be flat with almost vertical ends. At the other end
of the critical interval the two relevant roots $\alpha^\star_{3,4}=\pm \theta_p$
in row 5 of table \ref{t:1} of (\ref{critical}) approach $\arccos(x_0/L)$,
so we expect the limiting geometry to be a symmetrical, upward pointing wedge; here again we have a
mathematical kink joining two string segments with zero tension. 

While both the bifurcation diagram and the limiting geometries are identical
for the bridge and the wind-type loads, the non-smooth geometry inside the critical range is different.
In case of the bridge load coupled with gravity, the resultant of both loads
is vertical so, if the string is a straight line defined by one of the relevant roots $\alpha^\star_{3,4}=\pm \theta_p$
then identically vanishing tension $T(s) \equiv 0$ is possible, so \emph{any number} of mathematical kinks is admissible.
This is highly degenerated behaviour and it differs qualitatively from the other investigated load types.

In case of wind load coupled with gravity, the situation is different: here the resultant of the wind is
normal to the string while gravity is vertical, so tension will change in straight solutions monotonically.
This implies that only one end of a straight segment can have zero tension, so at most one kink is
admissible if the string consists of straight segments.

\subsubsection{Hydrostatic pressure combined with gravity}\label{hydro2}

If we combine hydrostatic pressure with gravity, the result at the lower
end of the critical parameter interval ($h/g=x_0/L$) is analogous to the upper
end of the previous case: the two relevant roots $\alpha^\star_{1,2}=\pm \theta_h$ (in row 7 of table \ref{t:1}) approach 
$\arccos(x_0/L)$ so here we can see a \em downward \rm pointing wedge
with a mathematical kink. If the load ratio exceeds the critical value, the
mathematical kink can not be sustained and we will immediately observe a physical kink
(stabilised by self-contact and bending),
joining two curved string segments.
The upper end of the critical interval ($h/g=1$) is less
trivial: here (\ref{critical}) has only one root, corresponding to the horizontal line.
Unlike the previous
case, here the vertical lines are not solutions so the boundary conditions can not
be satisfied by using straight lines corresponding to the relevant roots.
From this we can conclude that among smooth solutions a self-intersection must occur for $h/g>1$.
Since self-intersection is physically inadmissible, the string will develop self-contact.
Vertical segments of self-contact are in equilibrium and can be joined by non-straight
segments. Inside the critical range predicted by Corollary 3 we expect to see
smooth segments (convex from above) separated by vertical regions in self-contact.

The critical geometries and their order is illustrated in Figure \ref{mainfig}.

\begin{figure}[here]
%[!ht]
\begin{center}
\includegraphics[width=130 mm]{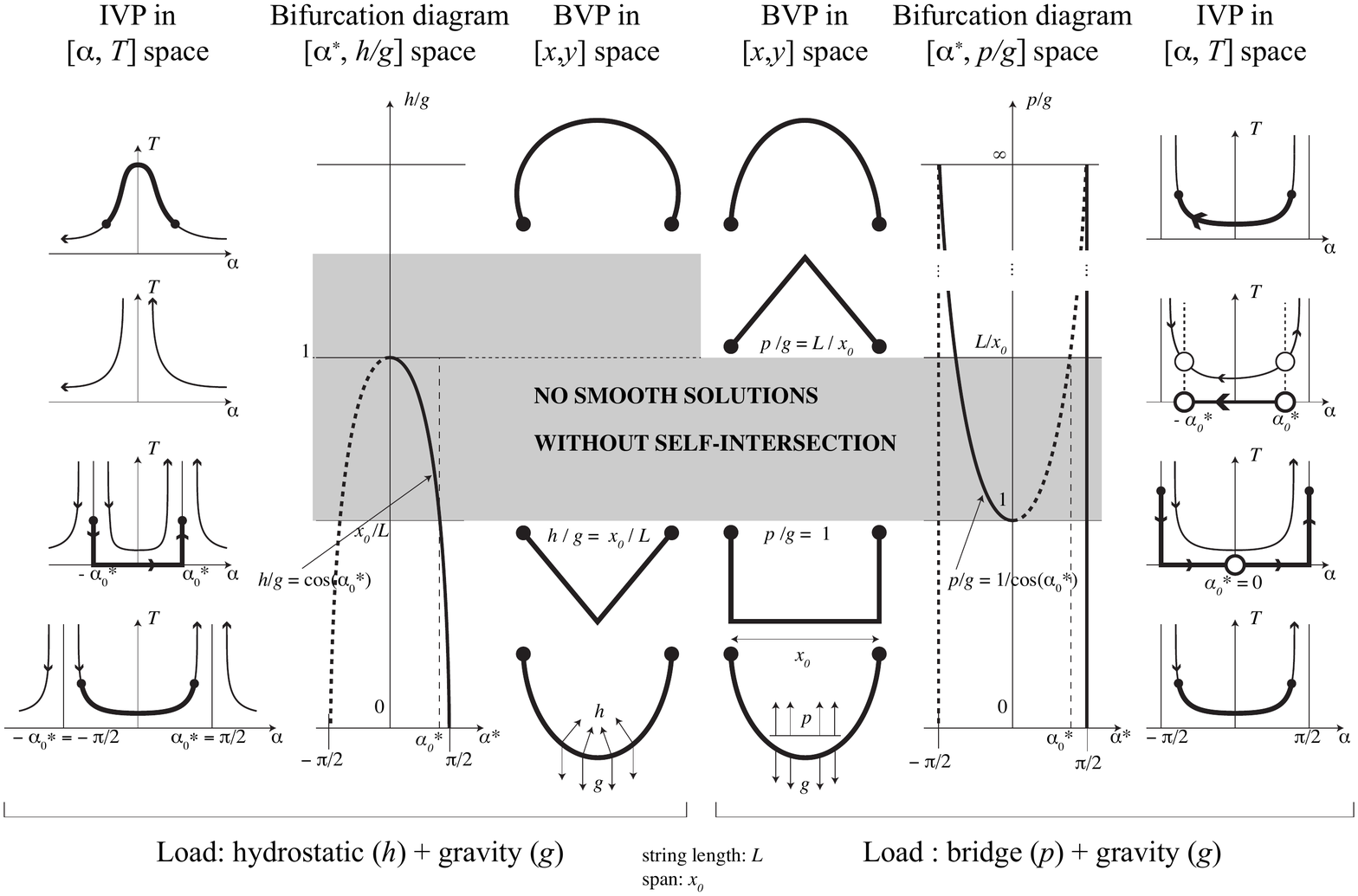}
\end{center}
\caption{Limiting geometries of the Boundary Value Problems for the combined gravity-bridge and gravity-hydrostatic pressure loads are shown in the two central (3rd and 4th) columns. The corresponding trajectories for the Initial Value Problems are shown in the first and last column. The bifurcation diagrams of the critical solutions are in the 2nd and 5th columns, with the stable and unstable solutions represented by solid and dashed lines, respectively. }\label{mainfig}
\end{figure}

\section{Numerical simulations}\label{s:numerical}

In order to study the behavior of the system in the regime that is in between the limiting solutions (gray zone in Figure~\ref{mainfig}), we carried out dynamical simulations. We modelled the flexible string using 100 beads, each connected to its two neighbours by very stiff springs with unit natural length and a steep repulsive interaction between the beads at close range to prevent the string crossing itself. The total energy of the system (without the external loads) was thus
\begin{equation}
E = \sum_i k(r_{i} - 1)^2 + \left(\frac{0.8}{r_{i}}\right)^5,
\end{equation}
where $r_{i} = |\mathbf{x}_i - \mathbf{x}_{i+1}|$ is the distance between the $i$th and the next particle along the string.  External forces due to the loads acted on the particles in addition, and an artificial velocity-dependent damping was introduced in order to drive the system to equilibrium if one existed. The first and last beads were kept fixed. The equations of motion were integrated with the velocity-Verlet scheme\cite{Frenkel} using a time step that was infrequently adjusted to keep the maximum displacement at each step under control.  
During the simulations, the gravitational force was kept constant, and the angle-dependent load (bridge load or hydrostatic pressure) was increased by constant increments from zero up to a large value, where the string resembled its limiting shape, and slowly decreased back down to zero again. In order to establish the existence of equilibrium shapes, the angle-dependent load was kept constant at each value until such time that the largest velocity component was below a threshold, or (allowing for the possibility that no equilibrium shape is stable) a certain time has elapsed since the last change in the load magnitude. 

Our results are illustrated in Figure \ref{numfig}, the arrangement is analogous to that of Figure \ref{mainfig}. For the hydrostatic pressure/gravity combination, self-contact was observed in the critical regime of the load ratio, while in the bridge load/gravity combination, no equilibrium solutions were observed here: the string displayed complex time-dependent dynamics. The dynamical behaviour of the string 
for the full range of the applied loads can be also viewed \cite{moviepg, moviehg}.

\begin{figure}[here]
%[!ht]
\begin{center}
\includegraphics[width=130 mm]{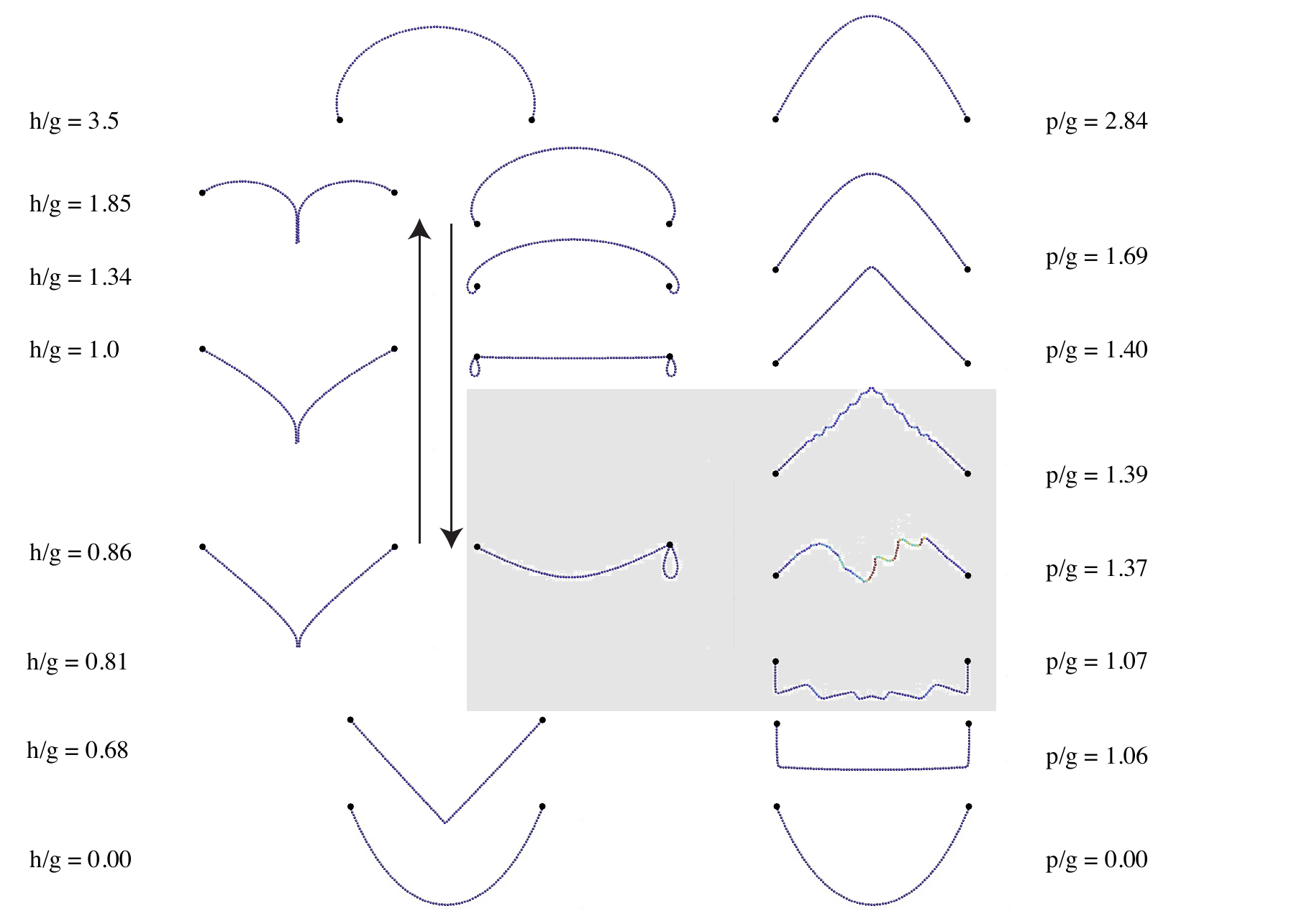}
\end{center}
\caption{Results of numerical simulations of the flexible string under combined load. The left two columns correspond to the hydrostatic pressure/gravity combination, with $h/g$ increasing (decreasing) in the left(middle) column. The right hand column corresponds to the bridge load/gravity combination. In the critical regime of load ratios, non-smooth and self-intersection equilibrium solutions are observed
in the hydrostatic pressure/gravity case, while for the bridge load/gravity combination we did not find any equilibrium solutions.}\label{numfig}
\end{figure}

\section{Experiments: gravity and hydrostatic pressure}\label{s:experiments}

We also carried out a simple table-top experiment for the combined gravity-hydrostatic pressure load case.
Our goal was to identify the limiting geometries described in subsection \ref{hydro2} and to compare the experiment
to the numerical simulations presented in section \ref{s:numerical}. Our equipment was an elongated rectangular
container with dimensions 110x25x35 cm on top of which we applied a rectangular  membrane of stress-free dimensions 110x40cm with an orthogonal mesh drawn for better shape identification.
We assumed that in the middle segment of the elongated rectangle the cross section of the membrane will approximate the shape of the planar string.
The internal pressure was regulated by hand with a pump and a valve and of course this system is conservative, as opposed to the nopn-conservative mathematical model. Nevertheless, we found
good qualitative agreemnt as it is illustrated in
Figure \ref{exptfig} showing a series of photographs taken
during a full load cycle. The wedge shape at the lower critical point as well as the relatively flat geometry prior to collapse are visible. 
The two segments in self-contact, joined by a physical kink, also resemble the numerical simulations.

\begin{figure}[here]
%[!ht]
\begin{center}
\includegraphics[width=13 mm]{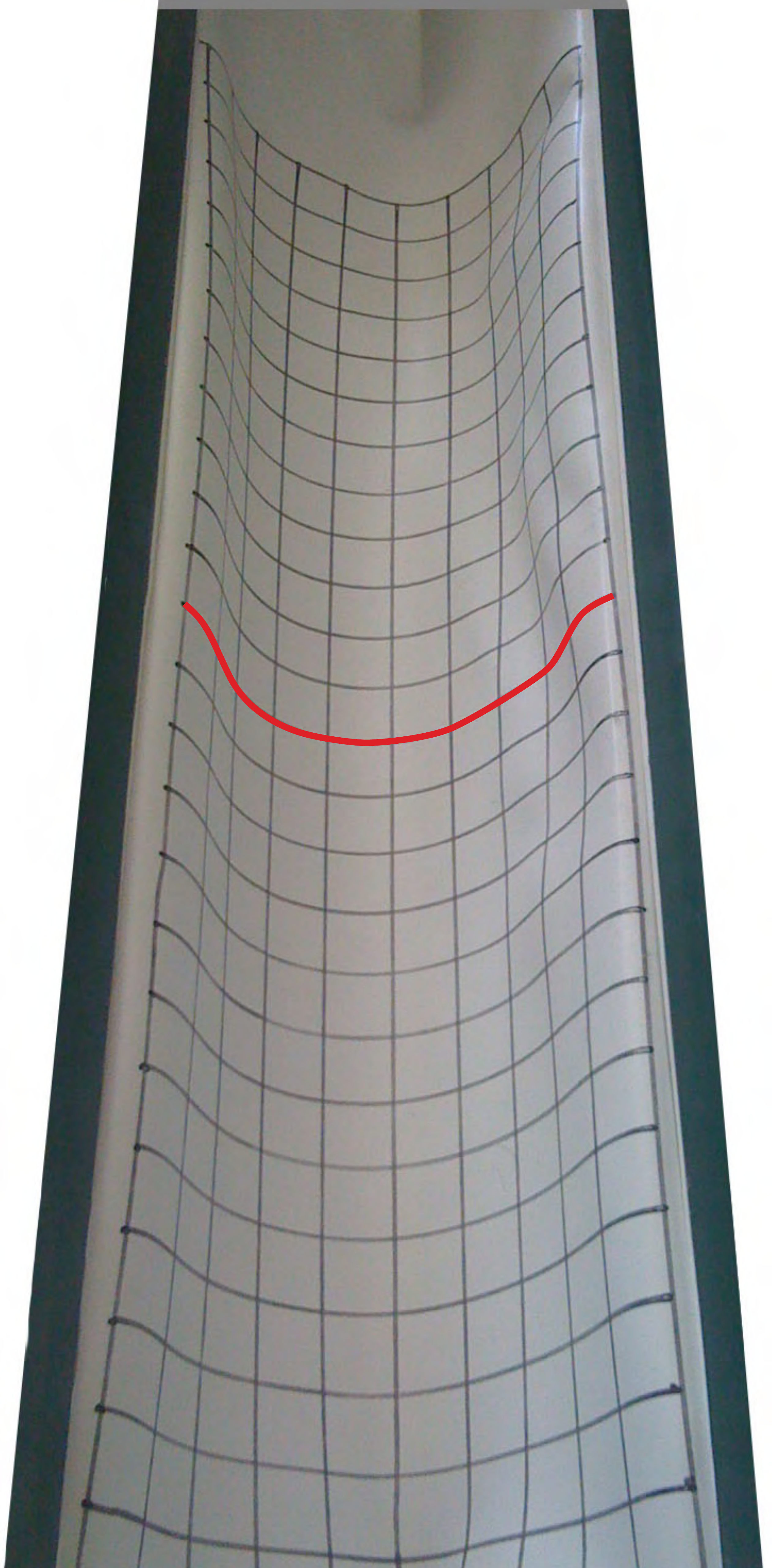}
\includegraphics[width=13 mm]{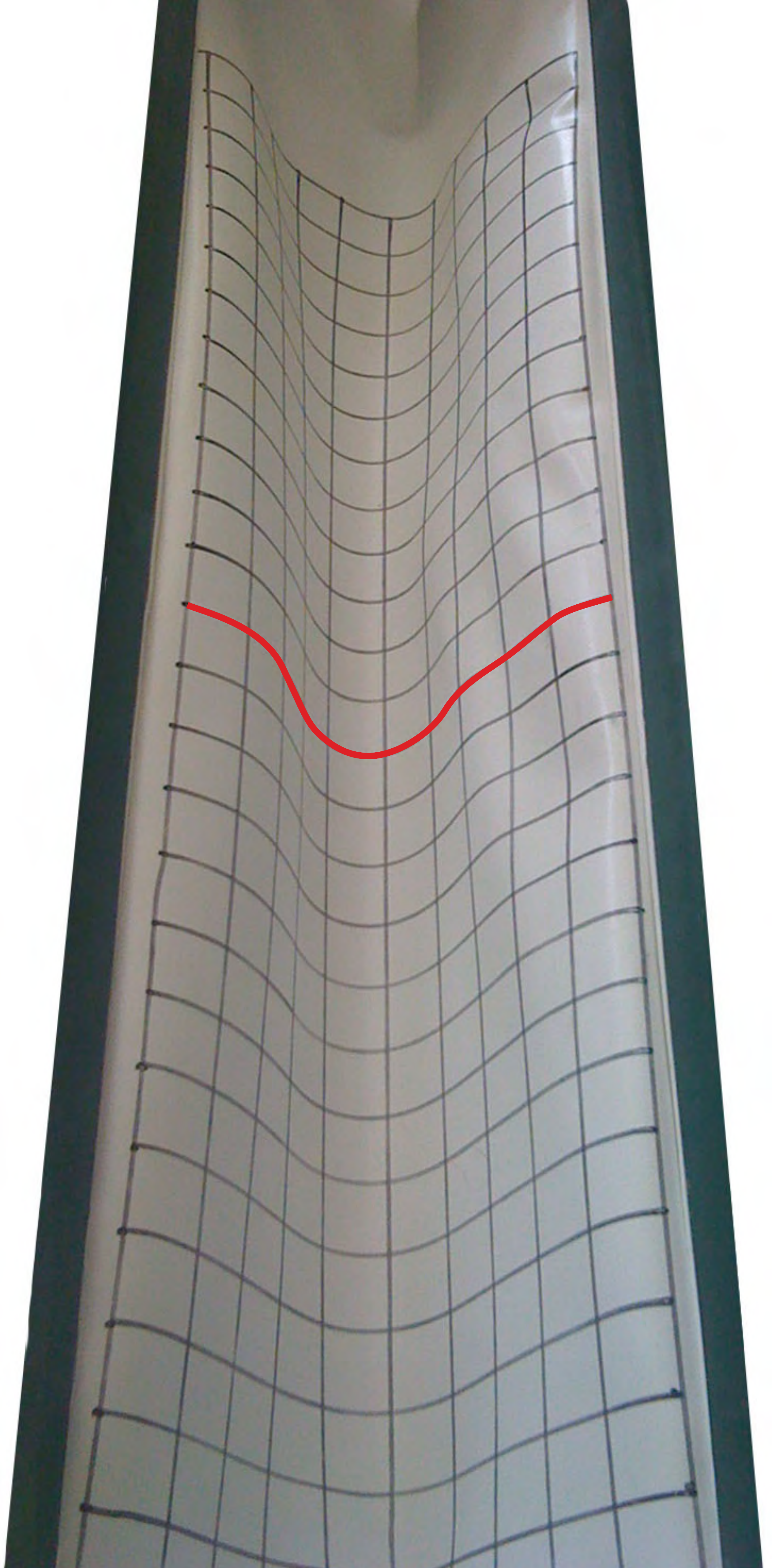}
\includegraphics[width=13 mm]{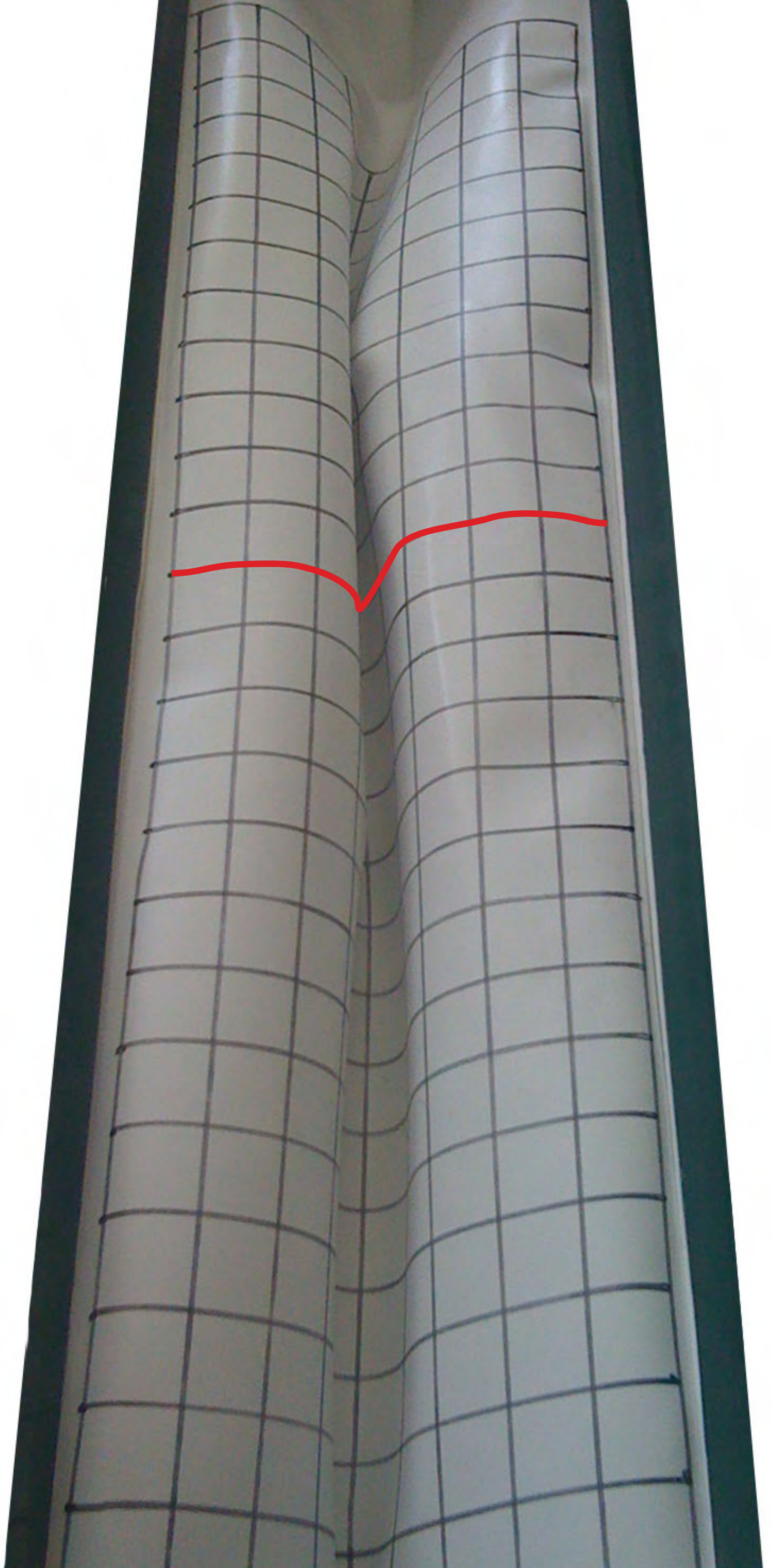}
\includegraphics[width=13 mm]{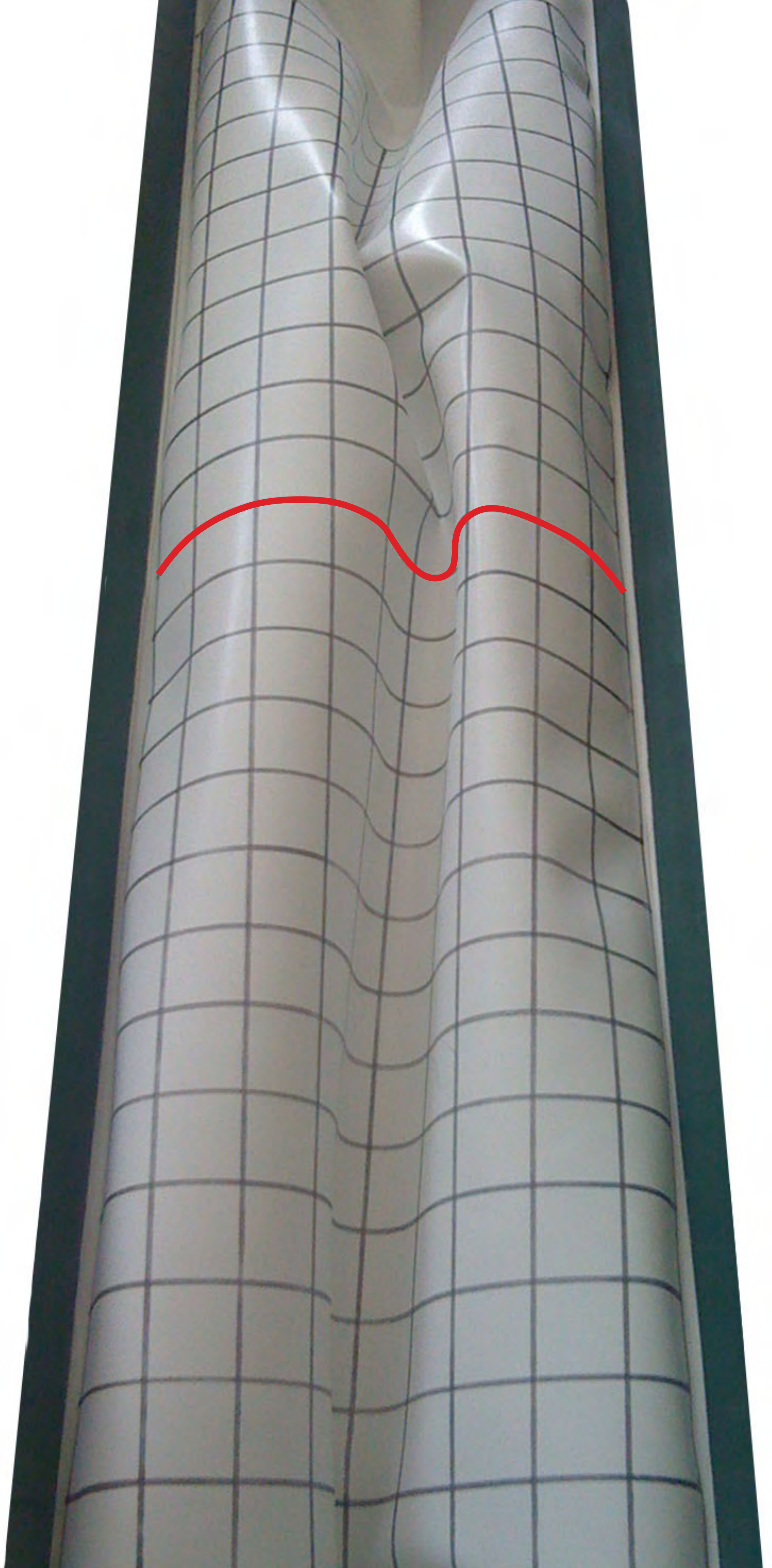}
\includegraphics[width=13 mm]{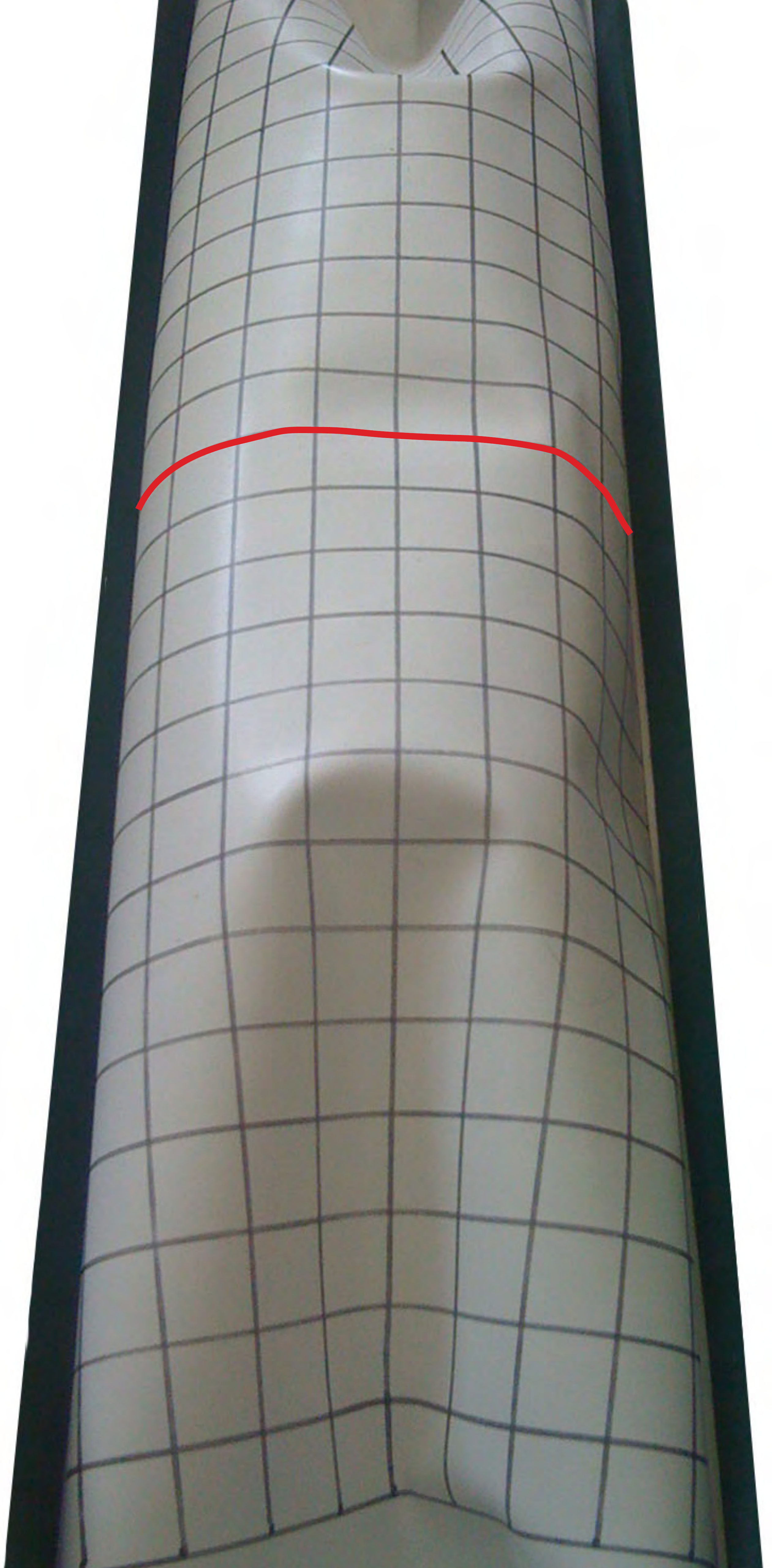}
\includegraphics[width=13 mm]{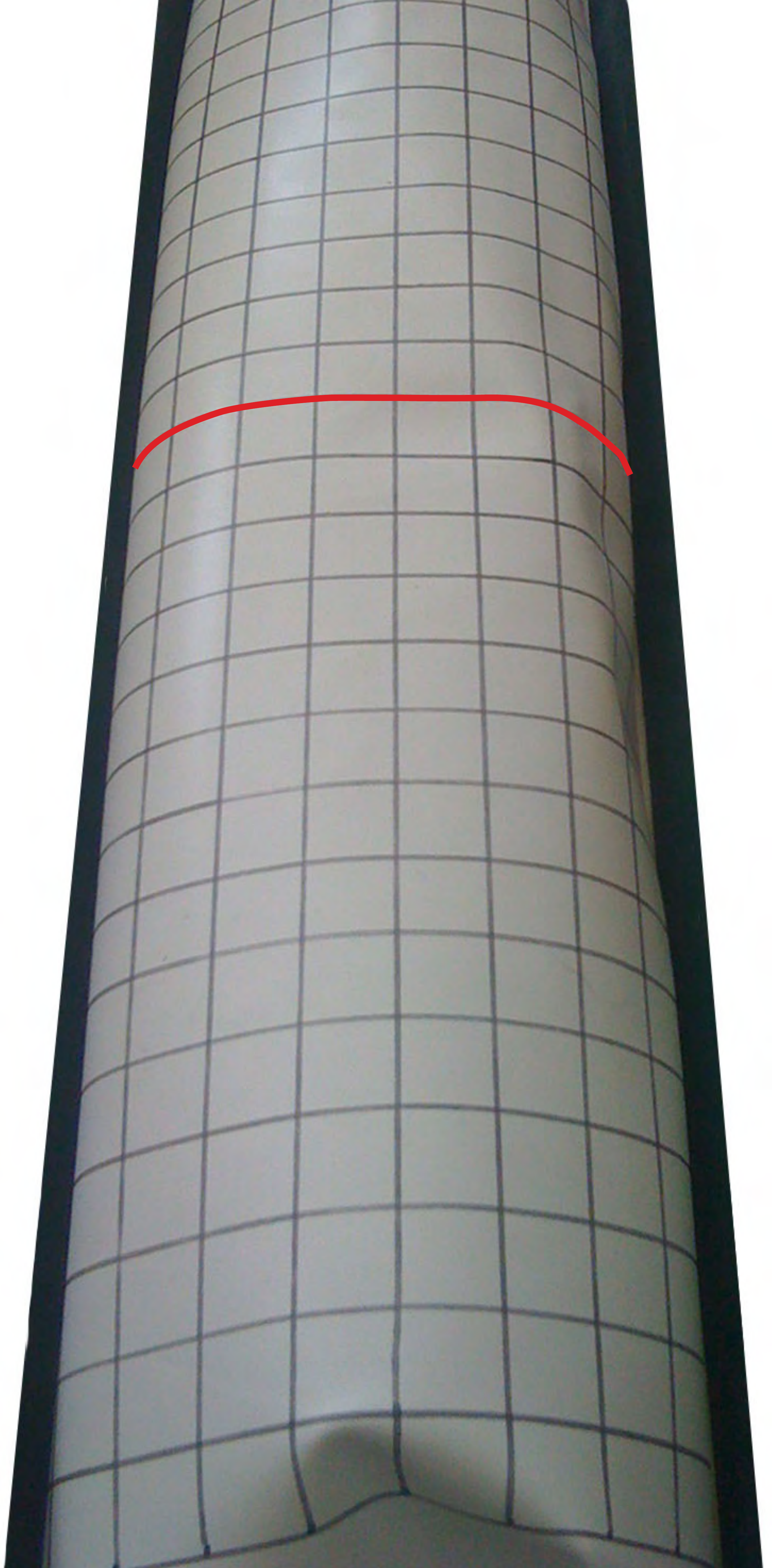}
\includegraphics[width=13 mm]{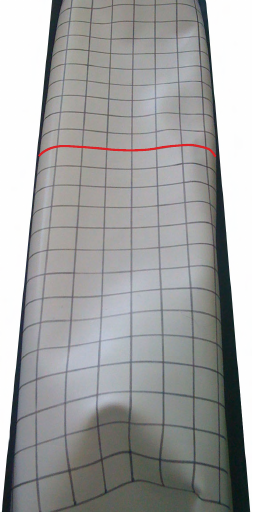}
\end{center}
\caption{A table-top experiments on the combination of hydrostatic pressure and gravity. The pressure is increasing from left to right from zero for the first six images, which correspond to the states shown in the left column (i.e increasing $h/g$ ratio) of Figure~\ref{numfig}. The last image shows the near-critical geometry as the load ratio is decreased from its maximum value.}\label{exptfig}
\end{figure}

\section{Conclusions and outlook}\label{s:conclusions}

\subsection{Summary of results}
We presented a theory about the global geometry of flexible strings when subjected to two
opposing smooth loads and we showed that under these loads the unilateral constraint in the
material that it does not support compressive loads may lead to non-smooth, spatially complex global geometry.
Under the assumptions that load intensities depend only on the slope of the string we showed
that the  key  to  the geometric collapse are the roots of the  critical equation (\ref{critical})
and we proved that the existence of \em relevant roots \rm (defined in  equations (\ref{bracket}), (\ref{delta}))
is the necessary and sufficient condition
for the nonexistence of smooth equilibrium configurations. We investigated four classical loads (gravity, bridge load,
Newtonian wind, hydrostatic pressure) and showed that if we couple any of the latter three with the first then
relevant roots appear and the geometry becomes non-smooth. While the behaviour under any one of these loads has been known
for centuries, this appears to be an interesting addition to the classical theory of strings.

Beyond proving the nonexistence of smooth solutions we also described the non-smooth regime.
In all three investigated cases the roots of (\ref{critical}) underwent a saddle-node
bifurcation as the load ratio was varied. One end of the critical regime with no smooth solution
is the bifurcation point, the other end is determined by the boundary conditions (length of string
and horizontal span). We determined the smooth limiting geometries at both (lower and upper) ends of the critical regime.
In case of bridge and wind load we found that the lower limiting geometry is a flat, rectangular shape
while the upper limiting geometry is an upward pointed, symmetrical wedge. In case of hydrostatic pressure
the lower limiting  geometry is a downward pointing wedge, while the upper limiting geometry can not be determined
as self-intersection occurs before the end of the critical regime can be reached. Nevertheless, the middle
portion of the string approaches the horizontal tangent.

By using numerical simulations we verified not only the above findings but also determined the non-smooth geometries
and the dynamical behaviour
inside the critical range of the load ratios. In case of hydrostatic pressure we performed a simple
table-top experiment which showed fair qualitative agreement both with theory and the numerical
simulations.

\subsection{Open questions and potential generalisations to membranes}
While our results apply to the geometry of membranes only under very special
conditions, the analogy between the models is apparent.
The equilibrium equations of thin membranes show a close analogy to (\ref{aa1}-\ref{aa2}) if no
in-plane shear is admitted:
\bea \label{mm1}
\bar \kappa &  =  & \frac{f_n}{T} \\
\label{mm2}
 T_i & =  & f_{t,i} \qquad i=1,2
\eea
where $\bar \kappa= \kappa_1 + \kappa_2$ is the sum of
the principal curvatures, $T_i$ refers to the partial derivatives of the tension $T$
and $f_{t,i}$ denote tangential loads \cite{Alexandrov},\cite{Rogers}. The main difference between the 2D problem (strings) and the 3D problem
(membranes) is that while strings are intrinsically flat (i.e. they do not have intrinsic curvature), the intrinsic
geometry of membranes is nontrivial and can be described by the Gaussian curvature.

In case of intrinsically flat membranes and simple boundary conditions (e.g. a rectangular membrane
supported along two parallel edges) our results directly apply and this was also confirmed
by our table-top experiments. In case of intrinsically curved membranes, equations (\ref{mm1})-(\ref{mm2}) have to be supplemented by the Gauss-Mainardi-Codazzi equations \cite{Alexandrov}. While the integrability of this coupled system 
has been investigated in \cite{Rogers} and is beyond the scope
of the current paper, we believe that qualitative features of our results still apply. In particular, we expect that
in case of intrinsically curved membranes the geometry of collapse will also include non-smooth shapes.

Our assumption on the slope-dependence of loads in (\ref{loads}) is not necessarily true in general, in particular,
strings with non-uniform mass distribution will have gravity loading which depends explicitly on the arc length.
We expect our results to be structurally stable and this
is supported both by the numerical simulation and by the table-top experiments. If assumption (\ref{loads}) is
not valid then one can not solve (\ref{aa1}-\ref{aa2}) by the simple Ansatz (\ref{ansatz}). However, we expect that for
loads slightly varying with arc length and, possibly with $x$ and $y$,
instead of critical solutions with \em constant \rm slope one will find critical solution
with \rm bounded \rm slopes.

\subsection{Applications}

Beyond giving insight into the global geometry of strings, our result can  also be used to predict
the approaching collapse of inflated tents in which  hydrostatic pressure is used to act against
gravity. In this case our model indicates that before collapse the middle of the roof will become
horizontal and  the curvature will decrease. The recent collapse of the inflated  Metrodome in Minneapolis (\cite{CNN}) shows that this appears
to be the case indeed. Here the collapse was caused by  snow accumulating on the tent.

It might be of interest to note that if tents were pressurised by a mechanism producing the
Newtonian wind-type load, then as the weight increased (or the pressure decreased)  the slope
towards the middle of the tent would \em increase \rm and this would prevent an incremental
collapse due to the accumulation of snow.

\section{Acknowledgements}
This work was supported by OTKA grant 104601 as well as the grant TAMOP - 4.2.2.B-10/1-2010-0009.

\end{document}